\documentclass[11pt]{iopart}

\usepackage{bm}
\usepackage{graphicx}

\usepackage{iopams}

\expandafter\let\csname equation*\endcsname\relax

\expandafter\let\csname endequation*\endcsname\relax

\usepackage{amsmath}

\usepackage{amsfonts}
\usepackage{amssymb}
\usepackage{amsthm}
\usepackage{amstext}
\usepackage{amsbsy}
\usepackage{amsopn}
\usepackage{amscd}
\usepackage{amsxtra}
\usepackage[french,english]{babel}
\usepackage{booktabs}
\usepackage[colorlinks=true,allcolors=blue]{hyperref}

\AtBeginDocument{
\heavyrulewidth=.08em
\lightrulewidth=.05em
\cmidrulewidth=.03em
\belowrulesep=.65ex
\belowbottomsep=0pt
\aboverulesep=.4ex
\abovetopsep=0pt
\cmidrulesep=\doublerulesep
\cmidrulekern=.5em
\defaultaddspace=.5em}

\newtheorem{proposition}{Proposition}

\newtheorem{conjecture}{Conjecture}

\newcommand{\Mc}[1]{\mathcal{#1}}
\newcommand{\setZ}{\mathbb{Z}}
\newcommand{\setN}{\mathbb{N}}

\newcommand{\setC}{\mathbb{C}}
\newcommand{\Id}{\mathbb{I}}

\newcommand{\bra}[1]{\langle #1 |}
\newcommand{\ket}[1]{| #1 \rangle }
\newcommand{\braket}[2]{\langle #1| #2 \rangle }

\begin{document}

\title{Selective and Robust Time-Optimal Rotations of Spin Systems}
\author{Quentin Ansel}
\address{Institut UTINAM, UMR 6213 CNRS-Universit\'e de Bourgogne-Franche-Comt\'e, Observatoire de Besançon, 41~bis Avenue de l’Observatoire, BP1615, 25010 Besançon cedex, France.}
\ead{quentin.ansel@univ-fcompte.fr}
\author{Steffen J. Glaser}
\address{
Munich Center for Quantum Science and Technology (MCQST), Schellingstr. 4, D-80799, Munich\\
Department of Chemistry,
Technical University of Munich, Lichtenbergstrasse 4, D-85747
Garching, Germany}
\author{Dominique Sugny}
\address{ Laboratoire Interdisciplinaire Carnot de Bourgogne (ICB), UMR 6303 CNRS-Universit\'e de Bourgogne- Franche-Comt\'e, 9 Av. A. Savary, BP 47 870, F-21078 DIJON Cedex, France}
\vspace{10pt}
\begin{indented}
\item[]\today
\end{indented}

\begin{abstract}
We study the selective and robust time-optimal rotation control of several spin-1/2 particles with different offset terms. For that purpose, the Pontryagin Maximum Principle is applied to a model of two spins, which is simple enough for analytic computations and sufficiently complex to describe inhomogeneity effects. We find that selective and robust controls are respectively described by singular and regular trajectories.
Using a geometric analysis combined with numerical simulations, we determine the optimal solutions of different control problems. Selective and robust controls can be derived analytically without numerical optimization. We show the optimality of several standard control mechanisms in Nuclear Magnetic Resonance, but new robust controls are also designed.
\end{abstract}


\maketitle

\section{Introduction}

In recent years, progress in quantum control (QC)~\cite{cat,past-present-future,dong,altafini-ticozzi,RMPsugny} has emerged through the introduction of appropriate and powerful tools coming from mathematical control theory~\cite{bonnard_optimal_2012,pontryaginbook,liberzon-book,brysonbook,leemarkusbook,dalessandro-book,kirk_optimal_2004,borzi-book,boscain-book,jurdjevic-book}. These developments have been recognized as an essential requirement for the future application of quantum technologies~\cite{cat,roadmapQT}. In this context, optimal control theory (OCT) has been successfully developed and applied since the eighties to become nowadays a standard tool in quantum physics~\cite{cat,past-present-future}. The practical
use of OCT is far from being trivial and each control problem has to be analyzed using either geometric or numerical methods for low or high dimensional systems, respectively~\cite{bonnard_optimal_2012,jurdjevic-book,boscain-book}. The different methods are based on the Pontryagin Maximum Principle (PMP) which is one of the main mathematical tools in control theory. The geometric approach is well adapted to
the description of ideal or simple quantum systems. This method
leads to a complete geometric understanding of the control
problem, from which, one can deduce the structure of the
optimal solution, a proof of the global optimality, and
the physical limits of a given process such as the minimum time to reach the target state. Such results can
be determined essentially analytically or at least with a
very high numerical precision. Issues of increasing difficulty have been recently solved for
closed~\cite{boscain-mason,alessandro2001,sugny10,hegerfeldt2013} and open quantum systems~\cite{lapert2010,bonnard2012,KhanejaPNAS,khanejathree}, but also
for unitary transformations~\cite{albertini_minimum_2014,garon2013,khaneja2001,khaneja2002}. From a numerical point of view, several optimization
algorithms have been developed in QC~\cite{cat,koch2012,KHANEJA2005,calarco2011}. They are able to account for experimental constraints, or
robustness and selectivity issues. In spite of their efficiency, such methods have some limitations. Only local optimal solutions can be derived and they do not provide the underlying mechanisms of the control process. For improving robustness or selectivity of non-adiabatic control pulses, a standard scenario consists in controlling an ensemble of systems which differ from the values of one or several constant parameters~\cite{k1,k2}. This approach has been widely explored in QC, mainly by OCT~\cite{KOBZAR2004,KOBZAR2005,KOBZAR2008,KOBZAR2012,SKINNER2012}, but also by learning algorithms~\cite{rabitz14,turinici19}. Recently, different studies have put forward some geometric properties of such control fields using OCT, composite pulses, shortcut to adiabaticity or specific parameterizations (see e.g. to mention a few~\cite{van_damme_time-optimal_2018,vandamme2017a,vandamme2017b,daems:2013,barnes2019,barnes2018,vitanov2014,Ruschhaupt_2012,Zeng_2018,linature}). Analytical expressions of the control pulses have been also derived.

We propose in this study to follow this direction, and to investigate the selective and robust time-optimal control of rotations on spin systems. Rotation corresponds in this study to SO(3)-transformation. In Nuclear Magnetic Resonance (NMR)~\cite{ernstbook,levittbook}, the control of an inhomogeneous ensemble of spin 1/2 particles with different offsets represents a benchmark example~\cite{silver1985,silvernature}.  This can be used as a building block for more complex processes. The mathematical analysis is guided by a simple control mechanism in NMR which involves two spins with different Larmor frequencies. The difference of frequency, noted $\Delta_1$, is called below by \textit{the offset}. In the rotating frame of the first spin, a rectangular pulse (with constant amplitude and phase) can rotate the first spin by an angle $\phi$ (usually $\phi = \pi$ or $\pi/2$), while rotating the second spin by multiples of $2\pi$. This rotation brings back the second spin to its initial position. For a rectangular pulse with amplitude $\omega$, and duration $T$, the smallest offset difference $\Delta_1$ which realizes this control task can be computed using the relation: $T \sqrt{\omega^2+\Delta_1^2} = 2\pi$~\cite{LEVITT198661}. In other words, the target state is reached for specific values of the offset parameter. This control mechanism is the starting point of the study and is closely related to the time-optimal processes derived from the PMP. As shown below, it corresponds to the time-optimal solution of the selective rotation of the two spins. Robust pulses can be constructed with a similar approach, but they require more complex control. A key point of the results obtained in this paper is that all the pulses can be expressed in a completely analytical way as a function of system parameters without numerical optimization. This analysis therefore provides a family of control solutions for the selective and robust control of the rotation of spin systems. As a byproduct of this study, we show the optimality of several standard control mechanisms in NMR and we design new control processes, which could be interesting for experimental applications.

This paper is organized as follows. In section~\ref{sec1}, the model system and the application of the PMP are presented. In section~\ref{sec:time opt selective pulses}, we compute and characterize time-optimal selective control fields. In section~\ref{sec:time opt robust pulses}, we consider the case of time-optimal robust controls. Conclusion and prospective views are given in section~\ref{sec:conclusion}. Technical details and additional control problems are presented in the appendices.

\section{Pontryagin Maximum Principle}
\label{sec1}

In this section, we present the model system and the application of the Pontryagin Maximum Principle (PMP) to the design of SO(3)- transformations. We characterize singular and regular trajectories of Pontryagin's Hamiltonian and we show that the system cannot switch from regular to singular extremals.

\subsection{Ensemble of SO(3)-transformations}
\label{sec:Definitions of robust/selective SO(3) transformations}

We consider an inhomogeneous ensemble of uncoupled spin- 1/2 particles with different offsets $\Delta$ in a given rotating frame. The offset is a frequency difference between the frequency of the spin and the one of the rotating frame. This later is defined implicitly by $\Delta=0$ for one spin of the ensemble. SO(3)- transformations are generated for a specific isochromat by the Bloch equation~\cite{k1,levittbook}:
\begin{equation}
\begin{split}
\frac{d \hat U(\Delta,t)}{dt} & = \left(\begin{array}{ccc}
0 & \Delta & -\omega_y \\
-\Delta & 0 & \omega_x(t) \\
\omega_y& -\omega_x(t) & 0
\end{array}  \right) \hat U(\Delta,t) \\
&= \left( \omega_x \hat \epsilon_x + \omega_y \hat \epsilon_y+ \Delta \hat \epsilon_z\right) \hat U(\Delta,t), \\
\hat U(\Delta,0) &= \hat \Id,\\
\end{split}
\label{eq:Bloch_equation}
\end{equation}
where $\omega_x$ and $\omega_y$ are two time-dependent control inputs with the constraint $\omega_x^2+\omega_y^2\leq \omega_0^2$, $\omega_0$ being the bound of the control field. $\Id$ denotes the identity matrix. In this study, we investigate the case $\omega_x(t) \in [-\omega_0,\omega_0]$ and $\omega_y=0$, for which analytic results can be derived. The symbols $\hat \epsilon_x$, $\hat \epsilon_y$ and $\hat \epsilon_z$ denote generators of the $\mathfrak{so}(3)$ algebra~\cite{hall_lie_2015,dalessandro-book}. The skew-symmetric matrices $\hat \epsilon_{x,y,z}$ verify the commutation relations $[\hat \epsilon_x , \hat \epsilon_y] = -\hat \epsilon_z$, $[\hat \epsilon_y , \hat \epsilon_z] = -\hat \epsilon_x$, and $[\hat \epsilon_z , \hat \epsilon_x] = -\hat \epsilon_y$ and $\hat \epsilon_a^\intercal=-\hat \epsilon_a$. They play the same role as Pauli matrices for the group $SU(2)$. Arbitrary dimensionless units are used throughout the paper.

The robustness of a control process can be defined either locally \cite{daems:2013,vitanov2014} or globally by considering a collection
of quantum systems~\cite{KOBZAR2004,KOBZAR2008}. The local robustness is generally defined from a cost function, $F(\Delta)$ which models the fidelity of the transformation with respect to the parameter $\Delta$ at the final control time. The aim is to cancel the first-order derivatives of $F$ around a specific value of $\Delta$, i.e. $\frac{\partial^n F}{\partial \Delta ^n}\vert_{\Delta = \Delta_0} = 0, n=1,2,...,n_{max}$. The global (or broadband) robustness is based on a discretization of the parameter space. This defines an ensemble of systems which differ by the value of the parameter.. The goal is then to control simultaneously each element of the ensemble toward the same target state. At the final time, we have $F(\Delta_{(n)}) = F_0$ for any offset $\Delta_{(n)}$ in the discretized interval.

In this paper, we combine these two approaches with the specificity that distinct target states are used to describe selective control. Following the control mechanism presented in the introduction, we consider two spins with offsets such that $\Mc C_\Delta = \{\Delta_{(0)} = 0, \Delta_{(1)}=\Delta_1\}$. The target states are defined by a rotation of angle $\phi$ around the axis $x$ for the offset of frequency $0$ (the spin in resonance), and the identity transformation for the offset $\Delta_1$:
\begin{equation}
\begin{split}
\hat U_{target}(0) &= e^{\phi \hat \epsilon_x}, \\
\hat U_{target} (\Delta_1) & = \hat \Id.
\end{split}
\label{eq:target_state}
\end{equation}
The set of target states is denoted $\Mc C_{\hat U}$. Note that any transformation generated by the control $\omega_x(t)$ is symmetric with respect to the sign of the offset, hence we assume $\Delta_1 >0$ without loss of generality. The generalization to more elaborated situations is briefly investigated in \ref{sec:localy robust pulse}. To quantify the robustness of a control field, we introduce the following function:
\begin{equation}
F(\Delta) = \Vert \hat U(\Delta,T,\omega_x) - \hat U_{target}(\Delta = 0) \Vert^2,
\label{eq:cost_function_2}
\end{equation}
where $\Vert\cdot\Vert$ is the Frobenius norm. The parametrization described here allows us to design selective or robust control fields. This idea is illustrated in Fig.~\ref{fig:example_fidelity}. From a qualitative point of view, the transformation is said to be selective if the curve around $\Delta=0$ is \emph{squeezed}, or equivalently if $\Delta_1$ is minimized (i.e. $\Delta_1 \rightarrow 0$) such that $F(\Delta_1)$ remains an extremum. Inversely, a robust transformation is achieved by taking a large offset $\Delta_1$ (i.e. $\Delta_1 \rightarrow \infty$), while keeping the curve as flat as possible around $\Delta=0$. Mathematically, a strong condition to generate a robust process is to nullify the first derivatives of the cost function $F$. However, we will see that optimal solutions could be determined without computing explicitly the different derivatives.
\begin{figure}[h]
\begin{center}
\includegraphics[width=\textwidth]{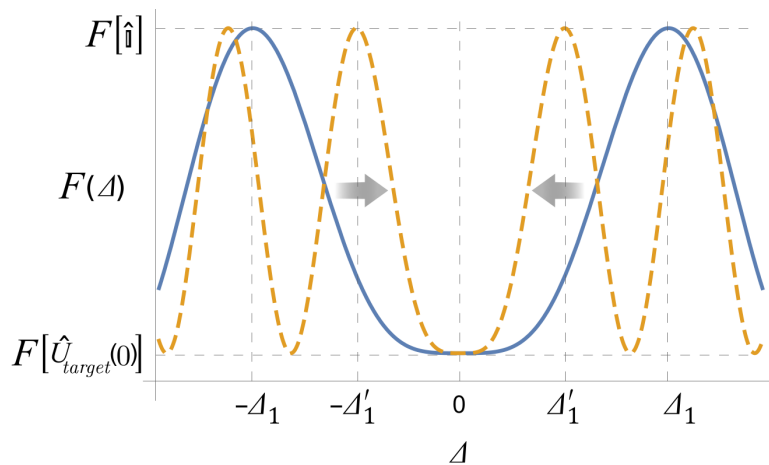}
\end{center}
\caption{(Color online) Schematic evolution of $F(\Delta)$ for two $SO(3)$- transformations with the conditions \eqref{eq:target_state} (blue line and orange dashed line). Horizontal gray dashed lines are used to indicate the target states. The notation $F[\hat U_f] = F(\Delta)$ with $\hat U(\Delta) = \hat U_f$ is introduced, in order to specify the transformation associated to the offset. The blue curve corresponds to a more robust control than the orange dashed one. To improve the selectivity, one has to "squeeze" the curve around $\Delta = 0$. This is symbolized by gray arrows.}
\label{fig:example_fidelity}
\end{figure}

A point to clarify is the use of the identity transformation as target state of the offset $\Delta_1$, instead of $e^{\phi \hat \epsilon_x}$. We emphasize that, for any control field, there exists $\Delta_1$ outside the interval of robustness that produces the identity transformation. The position of this offset depends non-trivially on the control field. The control problem is very difficult to solve with the usual approach, but computations are drastically simplified if we consider the condition $\hat U (\Delta_1) = \hat \Id$.

\subsection{Classification of optimal trajectories}
\label{sec:Pontryagin's Maximum Principle and the classification of optimal trajectories}

The Pontryagin Maximum Principle (PMP) describes the optimal control problem by using a Hamiltonian formalism~\cite{bonnard_optimal_2012,kirk_optimal_2004}. This approach allows us to derive a set of solutions which are candidates to optimality. In the case of a time-optimal problem, we define the following Pontryagin Hamiltonian:
\begin{equation}
H_p = \sum_{n=1}^N \braket{\hat P_n}{d_t \hat U_n }  = \bra{\hat P_n}\omega_x \hat \epsilon_x + \Delta_{(n)} \hat \epsilon_z \ket{\hat U_n}
\label{eq:H_p}
\end{equation}
where $\hat P_n$ are $\mathfrak{so}(3)$- matrices, called adjoint states of $\hat U_n$. We use the notation $\braket{A}{B}=\textrm{Tr}[A^\intercal B]$ for the matrix scalar product. We consider $N=2$ offsets, as described in section~\ref{sec:Definitions of robust/selective SO(3) transformations}, but this approach can be generalized to any set $\Mc C_\Delta = \{\Delta_{(0)},\Delta_{(1)},...,\Delta_{(N)} \}$. Examples with more than two offsets are considered in \ref{sec:localy robust pulse}. Equation~\eqref{eq:H_p} can be written as follows:
\begin{equation}
H_p=\sum_{n=1}^N [\omega_x l_n^x+\omega_y l_n^y+\Delta_{(n)} l_n^z],
\end{equation}
where $l_n^a=\langle \hat P_n | \hat \epsilon_a \hat U_n \rangle$, $a=x,y,z$. Introducing $l_a=\sum_n l_n^a$, $\vec l=(l_x,l_y,l_z)$ and $\vec{l}_n=(l_n^x,l_n^y,l_n^z)$, a compact expression of the Hamiltonian can be derived:
\begin{equation}
H_p=\omega_x l_x+\omega_y l_y+\sum_n \Delta_{(n)} l_n^z.
\end{equation}
The PMP states that the optimal trajectories satisfy Hamilton equations:
\begin{equation}
d_t \hat U_n (t) = \frac{\partial H_p }{ \partial \hat P_n(t)} ~~;~~ d_t \hat P_n (t) = - \frac{\partial H_p }{ \partial \hat U_n(t) },
\label{eq:Ham_eq}
\end{equation}
The efficiency of a control field is determined from the following cost functional:
\begin{equation}
C = \frac{1}{3N} \sum_{n=1}^N \Vert \hat U_n(T,\omega_x) - \hat U_{n,target} \Vert^2
\label{eq:cost_function}
\end{equation}
By construction, the optimal trajectory verifies $C=0$.
The application of the PMP leads to three possible types of trajectories~\cite{bonnard_optimal_2012,lapert_understanding_2013}:
\begin{enumerate}
\item Singular trajectories defined by the relation $\frac{\partial H_p}{\partial \omega_x} = 0$, and such that $|\omega_x(t)|\leq \omega_0$.
\item Regular trajectories for which $\omega_x = \omega_0\times  \text{sign}\left(\sum_{n=1}^2  \bra{\hat P_n}\hat \epsilon_x \ket{\hat U_n} \right)$.  The control field is a piecewise constant function which switches from $\pm \omega_0$ to $\mp \omega_0$ when the switching function  $l_x=\sum_{n=1}^2  \bra{\hat P_n}\hat \epsilon_x \ket{\hat U_n}$ is equal to 0.
\item Any concatenation of both solutions.
\end{enumerate}
A portion of the trajectory on an interval $I=[t_0,t_1]$ is called an arc. An arc is singular when the trajectory is singular for all $t \in I$. It is denoted with a $S$. Similarly, an arc is regular when the trajectory is regular for any $t \in I$. A regular arc for which $t_0$ and $t_1$ are two neighboring switching times is called \emph{a bang}, and it is denoted with a $B$. Concatenation of arcs is symbolized with a "-", as for example: $B-B$, $B-B-B$, $B-S$, etc.

We study below the different cases in order to select the corresponding optimal trajectory. This last step is performed in sections \ref{sec:time opt selective pulses} and \ref{sec:time opt robust pulses}, for respectively selective and robust transformations.\\
\noindent\textbf{Singular arcs:}\\
Singular arcs of this control problem are quite simple:
\begin{proposition}
In the case of two offsets $0$ and $\Delta_1$, singular arcs are given by constant controls of amplitude $|\omega_S| < \omega_0$.
\end{proposition}
\begin{proof}
First, we write \eqref{eq:H_p} as follows:
\[
H_P=\omega_x l_x+\Delta_1 l^z_1.
\]
Singular arcs $S$ cannot be directly derived from $H_P$ since they satisfy $l_x(t)=0$ in a non-zero time interval $[t_0,t_1]$. We deduce that $d_t{l_x}=d^2_t{l_x}=0$ in $[t_0,t_1]$. An explicit computation leads to:
\begin{equation}
\begin{cases}
l_x=l_0^x+l_1^x=0 \\
d_t{l_x}=-\Delta_1 l_1^y=0 \\
d^2_t{l_x}=-\Delta_1^2 l_1^x+\omega_S\Delta_1 l_1^z=0
\end{cases}
\label{eq:singular_conditions}
\end{equation}
where $\omega_S$ is a singular control field. Since $H_P=\Delta_1 l_1^z$ and $l_0^x$ are constants of the motion ($H_P$ is different from zero), we obtain that the singular control field is a constant function. An admissible field is obtained when the absolute value of the field amplitude is smaller than $\omega_0$.
\end{proof}
The trajectories generated by the singular field of amplitude $\omega_S$ and duration $T_S$, are given by the corresponding evolution operator:
\begin{equation}
\hat U(\Delta,T_S)=e^{T_S(\omega_S\hat \epsilon_x + \Delta \hat\epsilon_z)}=e^{T_S\sqrt{\omega_S^2 + \Delta^2}(n_x \epsilon_x + n_z \hat\epsilon_z)},
\end{equation}
where $\vec{n}=(n_x,n_z)$ is a unit vector, and $\Delta$ is an arbitrary offset. This transformation is a rotation of angle $\phi$ at resonance and the identity for the offset $\Delta_1$ if:
\begin{equation}
\begin{cases}
T_S\omega_S=\phi \\
T_S\sqrt{\omega_S^2+\Delta_1^2}=2k\pi,~k \in \setN.
\end{cases}
\end{equation}
The smallest offset solution of these equations is:
\begin{equation}
\Delta_1=\frac{\sqrt{4\pi^2-\phi^2}}{T_S}.
\label{eq:Delta_singular}
\end{equation}
As could be expected, the smaller the offset, the longer the control duration is. In the rest of the paper, we denote by $\Delta_0$ the offset associated with a constant regular control of amplitude $\omega_0$ and duration $T_0$.  This result is commonly used in NMR, and it can be heuristically deduced using Fourier transforms~\cite{bernstienbook}. From  a linear approximation of the Bloch equation~\eqref{eq:Bloch_equation} which is valid around the equilibrium point, we deduce that the control duration is the Fourier dual of the offset frequency. They are connected by the relation $\Delta  T = 2 \pi$. Hence, we expect that $\Delta \simeq 1/T$.\\
\noindent\textbf{Regular arcs:}\\
\label{subsec:regular arcs}
We now investigate in detail the structure of regular arcs. From Hamilton equations~\eqref{eq:Ham_eq}, we deduce the dynamics of the Hamiltonian lift $\vec{l}$ (which is a continuous function):
\begin{equation}
\label{eq:dl/dt}
d_t \left( \begin{array}{c}
l^x \\
l^y \\
l^z
\end{array} \right) = \left(\begin{array}{c}
-\Delta_1 l^y_1 \\
\Delta_1 l^x_1 - \omega_x l^z \\
\omega_x l^y
\end{array}  \right) ,
\end{equation}
where $\omega_x= \omega_0 \textrm{sign}(l^x)$. An example of regular control is displayed in figure~\ref{fig:structure_regular_control}.

\begin{figure}[h]
\centering
\includegraphics[scale=1]{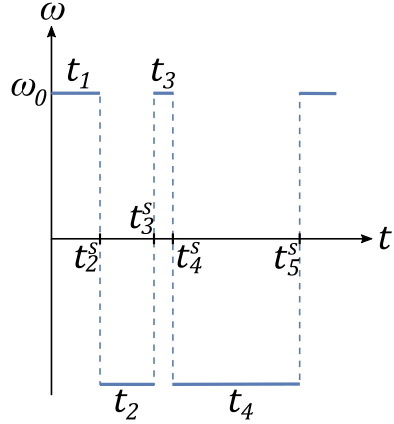}
\caption{Structure of a regular control field with several switchings.}
\label{fig:structure_regular_control}
\end{figure}
We introduce the switching times $t_1^s<t_2^s<\cdots <t_N^s$ of the control field and the bang duration $t_i = t_{i+1}^s - t_i^s$. Using \eqref{eq:dl/dt}, we arrive at:
\begin{equation}
l^y_1 (t_i^s +t) =    l^y_1(t_i^s) \cos(\Omega t)
+ \frac{l^x_1 (t_i^s) \Delta_1 - l^z_1(t_i^s) \omega_x }{\Omega}\sin(\Omega t),
\label{eq:solution_regular_pulse_ly1}
\end{equation}
with  $t\in [0,t_{i+1}^s-t_i^s]$, $\omega_x = \pm \omega_0$ and $\Omega = \sqrt{\omega_0^2 + \Delta_1 ^2}$. The integration of \eqref{eq:solution_regular_pulse_ly1} allows us to determine $l^x(t)$ as a function of $\vec l_1$:
\begin{equation}
\begin{split}
 l^x (t_i^s +t) =& l^x(t_i^s)- \frac{\Delta_1}{\Omega} [l^y_1(t_i^s) \sin(\Omega t) \\
&+ \frac{l^x_1 (t_i^s) \Delta_1 - l^z_1(t_i^s) \omega_x }{\Omega}(1-\cos(\Omega t))].
\end{split}
\label{l^x (t_i^s +t)}
\end{equation}
Note that a switching time $t_i^s$ is characterized by $l^x(t_i^s)=0$ and $l^y_1(t_i^s)\neq 0$. The general form of the control field as a function of the system state at time $t_i^s$ is deduced from \eqref{l^x (t_i^s +t)} by:
\begin{equation}
\omega_x(t_i^s+t) = \omega_0 \textrm{sign}[l^x(t_i^s+t)].
\label{eq:omega_x(t_i^s+t)}
\end{equation}
The control is entirely determined by $\vec l_n(t_i^s)$ at the switching time $t_i^s$, but a similar equation can be derived using the initial conditions $\vec l_n (t_0)$, $t_0$ arbitrary. In fact, it is not necessary to fix $\vec l_n$ at the initial time, we can fix the system state at any time, and integrate the dynamics forward or backward in time. It is also interesting to notice that the periodicity of \eqref{l^x (t_i^s +t)} implies the presence of past and future switchings, if there is at least one switching.  This parameterization is particularly useful to study concatenation of regular and singular arcs.

\noindent\textbf{Concatenation of regular and singular arcs}\\
We explore here the third set of extremal trajectories, i.e. the concatenation of regular and singular arcs.
\begin{proposition}
Let a regular arc $B$ defined at the switching time $t_i^s$. At the next switching time $t_{i+1}^s$, the system cannot switch to a singular arc $S$. Then, trajectories of the type $B-B-S-...$, $B-B-B-S-...$, etc, are not optimal.
\end{proposition}
This proposition could also hold for $B-S$ cases if the initial time is chosen during the second bang (hence, the first bang and the beginning of the second one are defined on $t<0$). It follows that only a trajectory of the type $B-S$ is actually used in the dynamics of the system.
\begin{proof}
We assume the existence of at least one switching. We start from the knowledge of the state at the switching time $t_i^s$, and we examine if the system can switch to a singular arc at the next switching (if it does not occur at infinity). Note that we can have $t_i^s<0$. From equations of section~\ref{subsec:regular arcs}, we deduce that $t_i$ is solution of:
\begin{equation}
\sin (\Omega  t_i) + A (1- \cos (\Omega  t_i)) =0,
\end{equation}
with,
\begin{equation}
A = \frac{l^x_1 (t_i^s) \Delta_1 - l^z_1(t_i^s) \omega_x }{ l^y_1(t_i^s) \Omega}.
\end{equation}
This is the unique solution because $l_1^y(t_i)=0$ and $A=0$ correspond to singular trajectories. A straightforward calculation leads to:
\begin{equation}
 t_i = \min_{k>0}\left[\frac{2}{\Omega}\left(k\pi, k \pi - \arctan \left(\frac{1}{A} \right)  \right) \right], t_i>0.
\label{eq:bang_duration}
\end{equation}
A possible switching from $B$ to $S$ arcs can be described by this analysis. Indeed, the singular set is defined by $l_x=d_t{l_x}=0$. We deduce that a time $t$ on this set verifies $l_x(t)=l_1^y(t)=0$, that is:
$$
\cos(\Omega t)+A\sin(\Omega t)=0,
$$
which leads to:
\begin{equation}\label{tsing}
t=\frac{1}{\Omega}[k\pi -\arctan(\frac{1}{A})].
\end{equation}
Equations~\eqref{eq:bang_duration} and \eqref{tsing} do not have joint solutions. Due to the time-reversal symmetry, and the freedom to choose the initial state at a switching time with $t_i^s<0$, the analysis covers the cases $B-B-S$, $B-B-B-S$, etc. (see \eqref{l^x (t_i^s +t)} and the comments below).
\end{proof}

More cautions have to be made for $S-B$ trajectories. Inserting conditions~\eqref{eq:singular_conditions} into \eqref{l^x (t_i^s +t)}, we obtain a possible solution characterized by:
\begin{equation}
l^x(t_i^s + t) = -\frac{\Delta_1^2}{\Omega ^2} l_1^x (t_i^s) \left( 1 - \frac{\omega_x}{\omega_S}\right) (1 - \cos(\Omega t))
\label{eq:singular_to_regular_lx}
\end{equation}
The duration of this arc is given by $t = 2 \pi / \omega = T_0$. The optimality of this trajectory is discussed in section~\ref{sec:time opt robust pulses}.

\section{Time-optimal selective transformations}
\label{sec:time opt selective pulses}%

The goal of this section is to determine the time-optimal solution in the selective case. Selective and robust controls can be seen as two opposite properties, therefore we respectively refer to the selective and robust problems when $\Delta_1 < \Delta_0$ and $\Delta_1 \geq \Delta_0$. The reason that motivates this choice becomes clearer in the next paragraphs.

The aim of the time-optimal selective control problem is to determine the shortest control field that generates the target state~\eqref{eq:target_state} when $\Delta_1 < \Delta_0$.
From section~\ref{sec:Pontryagin's Maximum Principle and the classification of optimal trajectories}, we know that $B-...-B-S...$ trajectories are not optimal. We present below analytic and numerical results supporting the following conjecture:
\begin{conjecture}
Time-optimal selective transformations are given by singular trajectories of Pontryagin Hamiltonian.
\label{conj_selective}
\end{conjecture}
%

\subsection{Analytic computations}
\label{sec:BB}

\begin{figure}[h]
\includegraphics[width=\textwidth]{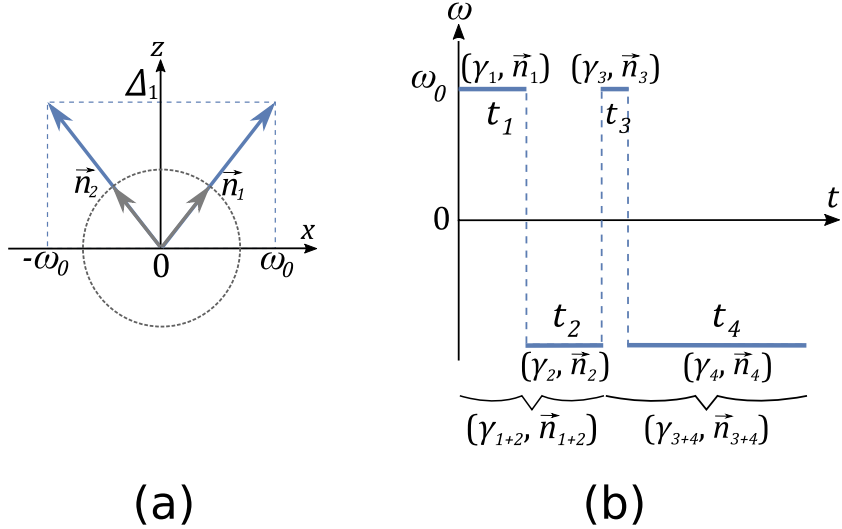}
\caption{(a) Two rotation axes involved in a regular control. (b) Example of a $B-B-B-B$ trajectory. Each bang is characterized by an angle $\gamma_i$ and a vector $\vec n_i$. The product of consecutive bangs gives effective angles $\gamma_{i+j+..}$ and effective vectors $\vec n_{i+j+...}$. In this example, $t_+ = t_1 + t_3$ and $t_- = t_2 + t_4$.}\label{fig:fig_proof}
\end{figure}

We show here the non-optimality of regular controls when the number of switchings is fixed. Since the number of possible cases to consider is infinite, we consider only the cases with one, two, and three switchings (i.e. $B-B$, $B-B-B$ and $B-B-B-B$ trajectories). We also prove the non-optimality of $S-B$ and $B-S$ controls, and the fact that a singular control can be asymptotically approximated by a regular solution with an infinite number of switchings. On the basis of these different results, we conjecture that singular trajectories are the time-optimal solutions (conjecture \ref{conj_selective}).

In the following, we use extensively the result established in \ref{sec:product of rotation}. We consider a regular trajectory of duration $T_S$, which has the same duration as a singular trajectory of amplitude $|\omega_S| < \omega_0$. Each bang is associated with a rotation angle $\gamma_i = t_i \sqrt{\omega_0^2 + \Delta_1 ^2}$ around an axis $\vec n_i$ (see figure~\ref{fig:fig_proof}). We introduce the notation $\gamma_{i+j+k+...}$ and $\vec n _{i+j+k+...}$ to denote effective rotation angles and rotation axes produced by the composition of several bangs, $i,j,k,...$. We also introduce the cumulative duration of bangs with positive and negative amplitudes, denoted respectively by $t_+$ and $t_-$. We have $t_+ + t_- = T_S$ and $t_+ - t_- = T_0$.

\noindent\textbf{The one-switching case:} Following \ref{sec:product of rotation}, and using the fact that the identity operator is parameterized by the unit quaternion, we deduce that $\gamma_{1+2} = 2 m \pi$ only if $\frac{\gamma_1 - \gamma_2}{2}= k \pi $ and $\frac{\gamma_1 + \gamma_2}{2}= n \pi$. These are the only solutions because $\vec n_1$ and $\vec n_2$ are not collinear (see figure~\ref{fig:fig_proof}). We arrive at:
\begin{align}
\label{eq:bang-bang-constr_1}
(t_1 - t_2) \sqrt{\omega_0^2 + \Delta_1^2} &= 2 k \pi ~,~k \geq 1 \\
\label{eq:bang-bang-constr_2}
(t_1 + t_2) \sqrt{\omega_0^2 + \Delta_1^2} &= 2 n \pi ~,~n \geq 1 .
\end{align}
Inserting the constraint $t_1-t_2 = T_0$ into \eqref{eq:bang-bang-constr_1} leads to:
\begin{equation}\label{eq_offset}
T_0 \sqrt{\omega_0^2 + \Delta_1^2} = 2 k \pi
\end{equation}
The smallest offset solution of \eqref{eq_offset} is by definition $\Delta_0$. Consequently, there is no offset $\Delta_1<\Delta_0$ solution of the control problem with a $B-B$ control of duration $T_S<T_0$.

\noindent\textbf{The two-switching case:} This computation is  more involved since the first two bangs can generate a rotation about an axis collinear to the last rotation axis. More precisely, we have to consider the following solutions:
\begin{align}
\vec n_{1+2} & \neq \pm \vec n_3 \\
\vec n_{1+2} & = \pm \vec n_3 .
\end{align}
In the first case, we have:
\begin{equation}
\left\lbrace \begin{array}{cc}
\gamma_{1+2}+\gamma_3 & = 2 k \pi \\
\gamma_{1+2}-\gamma_3 & =2 n \pi
\end{array} \right.
~ \Rightarrow ~
\left\lbrace \begin{array}{cc}
\gamma_{1+2} & = 2 k' \pi \\
\gamma_3 & =2 n' \pi
\end{array} \right.
\end{equation}
The condition $\gamma_{1+2} = 2 k' \pi$ is similar to the one-switching case. Then we can add \eqref{eq:bang-bang-constr_1} and $\gamma_3 = 2 n' \pi$ in order to obtain:
\begin{equation}
\label{eq:bang-bang-bang_constr_1}
\begin{split}
(t_1-t_2+t_3)\sqrt{\omega_0^2 + \Delta_1^2} &= 2 (k'+n')\pi \\
\Rightarrow T_0 \sqrt{\omega_0^2 + \Delta_1^2} &= 2 l \pi .
\end{split}
\end{equation}
The smallest offset is obtained for $l=1$, which corresponds to $\Delta_0$. It remains to analyze the situation in which the rotation axes are collinear. An explicit computation of the product of evolution operators gives:
\begin{eqnarray*}
\textrm{Tr}[\hat U_{1+2+3}]&=&\frac{1}{4} [8 \cos (2 \theta ) \sin^2\left(\frac{\gamma_2}{2}\right)  \sin ^2\left(\frac{\gamma_1+\gamma_3}{2}\right) \\
& &-8 \cos (\theta )\sin (\gamma_2)  \sin (\gamma_1+\gamma_3) \\
& & +3 \cos (\gamma_1-\gamma_2+\gamma_3)+3 \cos (\gamma_1+\gamma_2+\gamma_3) \\
& &  +2 \cos (\gamma_1+\gamma_3)+2 \cos (\gamma_2)+2]
\end{eqnarray*}
A necessary condition to obtain $\textrm{Tr}[\hat U_{1+2+3}]=3$ is: $\gamma_2 = 2k\pi$. The same result can be deduced by computing the real part of the product of three unit quaternions:

\begin{equation*}
\begin{split}
Re[\mathbf{q}_{1+2+3}] =& \cos (\gamma_1) \left(\cos (\gamma_2) \cos (\gamma_3)-\sin \left(\frac{\gamma_2}{2}\right) \sin \left(\frac{\gamma_3}{2}\right) \cos (\theta )\right) \\
& -\sin \left(\frac{\gamma_1}{2}\right) \left(\sin
   \left(\frac{\gamma_2}{2}\right) \cos (\gamma_3) \cos (\theta )+\cos (\gamma_2) \sin \left(\frac{\gamma_3}{2}\right)\right)
\end{split}
\end{equation*}
The composition rules of the two other angles give: $\gamma_1 + \gamma_3 = 2 n \pi$. We can now use $T_0=t_+-t_-$ to determine:
\begin{equation}
\begin{split}
(t_\pm - t_\mp)\sqrt{\omega_0^2 + \Delta_1^2} &= \pm 2 (l-m)\pi \\
\Rightarrow T_0 \sqrt{\omega_0^2 + \Delta_1^2} &= 2 n \pi \\
\Rightarrow \min \Delta_1 = \Delta_0 &
\end{split}
\end{equation}
\noindent\textbf{The three-switching case:} We decompose the control into two different parts and we study the rotations 1+2 and 3+4 as single blocks. The conditions to produce the identity are:
\begin{equation}
\vec n_{1+2} \neq \pm \vec n_{3+4}  \Rightarrow \left\lbrace \begin{array}{c}
\gamma_{1+2} = 2 k \pi \\
\gamma_{3+4} = 2 n \pi
\end{array} \right.
\end{equation}
\begin{equation}
\vec n_{1+2} = \pm \vec n_{3+4}  \Rightarrow \gamma_{1+2} + \gamma_{3+4} = 2m\pi .
\end{equation}
For the first case, we proceed as in \eqref{eq:bang-bang-bang_constr_1}, and we prove that there is no solution. Therefore, as for the two-switching case, the only non-trivial situation corresponds to collinear axes. From~\ref{sec:product of rotation}, we deduce that two products of two rotations generate an effective rotation around the same axis only if the angles are equal modulo $2k\pi$ (up to a sign, but here, all angles are positive). Therefore, we arrive at:
\begin{equation}
\gamma_{1} = \gamma_3 ~;~\gamma_2 = \gamma_4
\end{equation}
and
\begin{equation}
\gamma_{1+2} = \gamma_{3+4} = \pi .
\end{equation}
This is possible only if:
\begin{align}
\gamma_1 + \gamma_2 &= (2k+1)\pi \\
\gamma_1 - \gamma_2 &= (2n+1)\pi ,
\end{align}
and we deduce that:
\begin{equation}
(t_+ - t_-) \sqrt{\omega_0^2 + \Delta_1^2} = 2 m \pi .
\end{equation}
Finally, we recover the same result as in the other cases.

\noindent\textbf{$S-B$ trajectory:} We proceed similarly as in the previous situations, except that we fix the offset, and we compare the control duration $T_S$ of a singular control with the duration $T = t_1 + t_2$ of a $S-B$ control. Here, $t_1$ and $t_2$ are the respective durations of the singular and regular arcs. We set $T_S = \omega_S \phi$, $T_S\sqrt{\omega_S^2 + \Delta_1^2} = 2 \pi$ to characterize the singular trajectory, and $t_1 \omega_S + t_2 \omega_0 = \phi$, $\gamma_{1+2} = 2\pi$ for the regular one. Using $t_1 \omega_S + t_2 \omega_0 = T_S\omega_S$, we obtain that $T<T_S$, because $\omega_S < \omega_0$, but the second constraint has to be taken into account. Rotation axes associated with each part of the $S-B$ control are not collinear, and as usual $\gamma_1 \pm \gamma_2 = 2 k_\pm \pi$. This equation leads to:
\begin{align}
t_1 \sqrt{\omega_S ^2 + \Delta_1^2} &= 2 k \pi \\
t_2 \sqrt{\omega_0 ^2 + \Delta_1^2} &= 2 n \pi.
\end{align}
Inserting these equations into the condition at resonance gives:
\begin{align}
& k \pi \frac{\omega_S}{\sqrt{\omega_S ^2 + \Delta_1^2}} + n \pi \frac{\omega_0}{\sqrt{\omega_0 ^2 + \Delta_1^2}} = 2 \pi \frac{\omega_S}{\sqrt{\omega_S ^2 + \Delta_1^2}} \\
\Rightarrow & \frac{2 - k}{n} = \sqrt{\frac{\omega_0^2}{\omega_S^2}. \frac{\omega_S^2 + \Delta_1^2}{\omega_0^2 + \Delta_1^2}}
\end{align}
We have $ \sqrt{\frac{\omega_0^2}{\omega_S^2}. \frac{\omega_S^2 + \Delta_1^2}{\omega_0^2 + \Delta_1^2}} >1$ because $\omega_0 > \omega_S$, hence the only possible solutions are $k=0$ or $(n=0, k = 2)$. The first case is non-physical while the second corresponds to a single singular arc. These calculations are extended easily to the case $B-S$ by permuting the definition of $t_1$ and $t_2$. Therefore, the concatenation of a singular arc with a regular one is not optimal. 

We can also continue calculations with $B-S-B$ trajectories. They are optimal in some selective state-to-state transfers of spin systems \cite{van_damme_time-optimal_2018}. For conciseness, we do not present this case in this paper, but we obtain a similar conclusion as $S-B$ and $B-B-...$ trajectories. We finally obtain that all admissible trajectories are longer than the singular one.

\noindent\textbf{Continuum limit:} The computations with more than three switchings become quickly arduous due to the large amount of possible new cases to consider. However, it seems that we always find the same conditions, giving $\Delta_0$ as the smallest offset to realize the identity transformation with regular arcs. Another argument supporting  this conjecture is provided by an analysis of the continuum limit. This result is known in the context of average Hamiltonian theory~\cite{brinkmann_introduction_2018}. We assume that the interval $[0,T_S]$ is divided in $M$ equal parts. We have:
$$
\mathbb{T}e^{\int_{\sigma_i} dt(\omega_x(t) \hat \epsilon_x + \Delta_1 \hat \epsilon_z)}= e^{T_S\bar \omega_x \hat \epsilon_x / M + o(T_S^2/M^2)} e^{T_S\Delta_1 \hat \epsilon_z / M + o(T_S^2/M^2)},
$$
with $\bar \omega_x$ the average of the control field in the interval $\sigma_i$ and $\mathbb{T}$ the time ordering operator. In the limit $M \rightarrow \infty$, we can use the Trotter formula~\cite{suzuki_generalized_1976}, which gives $\hat U (T_S) = e^{T_S(\omega_S \hat \epsilon_x + \Delta_1 \hat \epsilon_z)}$. This is the solution given by singular controls. A regular trajectory with a large number of switchings can therefore approximate a singular trajectory.

\subsection{Numerical computations}
\label{sec:numerical verification}

\begin{figure}[t]
\includegraphics[width=\textwidth]{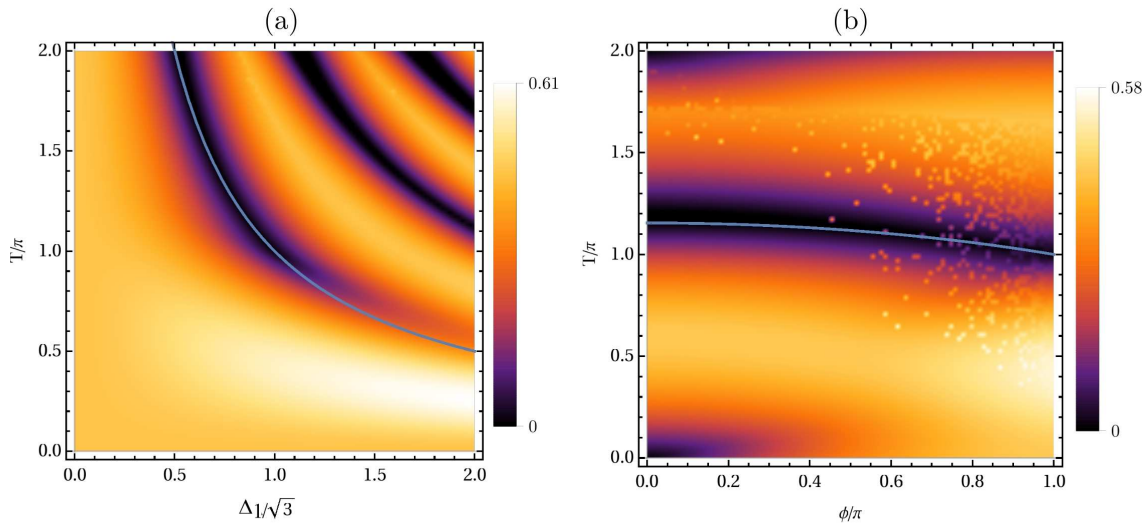}
\caption{(Color online) Panel (a): Plot of the cost $C$ as a function of the control duration $T$ and the offset $\Delta_1$, for a selective rotation of angle $\phi = \pi$ and $\omega_0=1$. Panel (b): Plot of $C$ as a function of $T$ and the angle of rotation $\phi $. We set $\omega_0=1$ and $\Delta_1 = \sqrt{3}$. In the two cases, each point of the contour plot corresponds to a numerical optimization (see the text for details). Blue solid lines are defined by the equation \eqref{eq:Delta_singular}.}
\label{fig:landscape_shooting_GRAPE}
\end{figure}

We perform numerical simulations in order to verify conjecture~\ref{conj_selective}. The goal is to solve the problem by "brute force" minimization of $C$ (see \eqref{eq:cost_function}) with the GRAPE algorithm~\cite{KHANEJA2005}. Here, there is no assumption about the structure of the control field, which is not constrained to be singular or regular. It is a piecewise constant function such that $|\omega_x (t)| \leq \omega_0 $. GRAPE is an iterative algorithm which improves step by step a test solution using gradient descent. By construction, such control fields do not verify $C=0$, i.e. they do not reach the target state exactly, but a state as close as possible to the target state. Transversality conditions at the final time (which is fixed) are taken into account in the numerical algorithm. Then, the minimum time of the process is given by the shortest time for which the target state is reached with a sufficiently small precision.

In order to highlight the relation between control duration and offset selectivity, several optimizations at fixed final time have been performed for different offset values. Results are presented in figure~\ref{fig:landscape_shooting_GRAPE}(a). We observe that the curve of equation $T(\Delta_1)= \frac{1}{\Delta_1}\sqrt{4\pi^2 - \phi^2}$ defines the global minimum of  $C$, and that a better selectivity cannot be achieved. This curve characterizes singular solutions (see \eqref{eq:Delta_singular}). For $\Delta_1 > \sqrt{3}$ (the value of $\Delta_0$ in figure~\ref{fig:landscape_shooting_GRAPE}), the optimization algorithm does not find good solutions around the singular trajectory. This is due to the limited amplitude of the control field. Longer controls with more complicated structures are required, as emphasized by the other black areas of figure~\ref{fig:landscape_shooting_GRAPE}(a). For these controls, the numerical solution is close to a regular control (bang-bang structure).

The general character of these results is verified by several optimizations with different rotation angles and control duration (the offsets are fixed). The corresponding results are presented in figure~\ref{fig:landscape_shooting_GRAPE}(b), in which  we also observe the optimality of singular constant controls. These numerical simulations are in agreement with the previous analytic results, and they support our conjecture.

\section{Time-optimal robust transformations}
\label{sec:time opt robust pulses}

In this section, we focus on the analytic computation of time-optimal robust controls. We also compare these results with some published solutions~\cite{vandamme2017a}. The goal of the time-optimal robust control problem is to determine the shortest control field, which generates the target states~\eqref{eq:target_state} when $\Delta_1 \geq \Delta_0$, with the smallest variations of $F(\Delta)$ near the offset $\Delta = 0$: $\frac{\partial^n F}{\partial \Delta ^n}\vert_{\Delta = 0} = 0, n=1,2,...$. Since there is no singular control associated to $\Delta_1 \geq \Delta_0$, we consider only regular trajectories.

The general expression for the evolution operator with a regular control is:
\begin{equation}
\hat U = \mathbb{T} \prod_{j=1}^{N_p} e^{(\omega_j \hat \epsilon_x + \Delta \hat \epsilon_z)t_j },
\end{equation}
where $\omega_j = \pm \omega_0$. We determine precisely the admissible set of values $t_j$ generating the target states for a fixed number of switchings. For that purpose, we use the computation already performed in section~\ref{sec:BB}. We recall the results here. We have respectively:
\begin{align}
(t_1 - t_2) \sqrt{\omega_0^2 + \Delta_1^{2}}&=2 k \pi & k\geq 1 \\
(t_1 + t_2) \sqrt{\omega_0^2 + \Delta_1^2}&=2 n \pi & n\geq 1,
\end{align}
for the one-switching case, and
\begin{equation}
(t_1 - t_2 + t_3) \sqrt{\omega_0^2 + \Delta_1^2}=2 k \pi ~;~ k\geq 1
\end{equation}
or
\begin{align}
t_2 \sqrt{\omega_0^2 + \Delta_1^2}&=2 n \pi & n\geq 1 \\
(t_1 + t_3) \sqrt{\omega_0^2 + \Delta_1^2}&=2 l \pi & l\geq 1,
\end{align}
for two switchings. For more switchings, we refer to section~\ref{sec:BB}.
By inserting $\sum_{j=1}^{N_p} t_j \omega_j = \phi$ in the conditions above, we determine the following structure for optimal regular controls. In the one-switching case, we have:
\begin{equation}
t_1 = t_2 + \frac{\phi}{\omega_0} ~ ; ~ t_2 = \frac{\phi}{2 \omega_0}\left(\frac{n}{k}-1\right)
\end{equation}
and for two switchings:
\begin{equation}
t_1 + t_3 = \frac{\phi}{\omega_0}\left(\frac{n}{k}-1\right) ~;~ t_1 = \alpha t_3 ~;~ t_2 =  (t_1+ t_3)+ \phi/\omega_0
\label{eq:param:opt_robust_2_switchings}
\end{equation}
Other solutions exist in the two-switching case, but they cannot be time-optimal, because of their long duration. The condition $\omega_x(0) = - \omega_0$ is used to determine \eqref{eq:param:opt_robust_2_switchings}.

To summarize, in the first situation, the optimization consists of finding two integers, $(n,k)$ and in the second case, the goal is to determine two integers, and one real number $\alpha$. By plotting the cost function for the first values of $(n,k,\alpha)$, we deduce which solution is optimal, as illustrated in figure~\ref{fig:compilation_surface}. As an example, for $\phi =\pi/2$, the solution $n=k=1$, $\alpha=1$ is the most robust one, while the solution $n=3$ and $k=2$, $\alpha = 1$ is better for $\phi =\pi$. We could also compute the derivatives of $F(\Delta)$ near $\Delta=0$, but generally, a simple observation is sufficient to find a solution.

\begin{figure}[h]
\includegraphics[width=\textwidth]{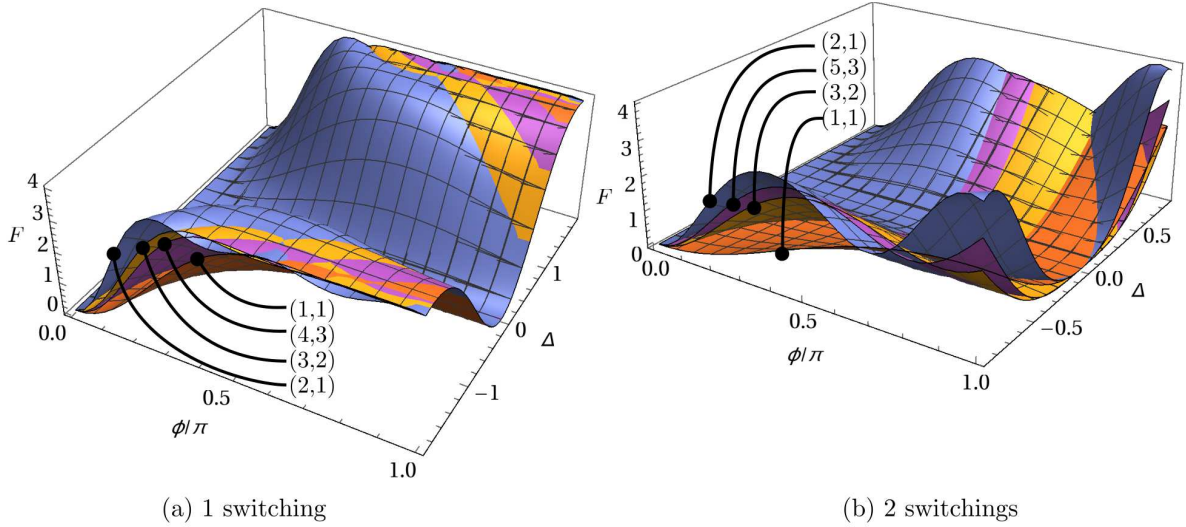}
\caption{(Color online) Plot of the cost $F$ for one (panel (a)) and two switchings (panel (b)) as a function of $\phi$ and $\Delta$ for different values of $(n,k)$. In the two-switching case, we set $\alpha=1$.}
\label{fig:compilation_surface}
\end{figure}

The computation of derivatives is generally necessary for the evaluation of the optimal value of $\alpha$. Further analytic calculations are possible if $\phi$ is fixed. Here, we investigate the case $\phi = \pi$. The optimal value is computed with the solution of $\frac{\partial^2 F}{\partial \Delta^2}(0)$ (the first derivative is always zero). A straightforward (but lengthy) computation gives:
\begin{align}
\frac{\partial^2 F}{\partial \Delta^2}(0) &= \left[ \frac{\partial^2 }{\partial \Delta^2} \Vert e^{-\pi \hat \epsilon_x} - \hat U(\Delta) \Vert ^2\right]_{\Delta = 0} \\
&=16 (3 - 2 \cos(\omega_0 t_1) - 2 \cos(\omega_0 t_3) + 2 \cos(\omega_0 (t1 + t3))) \\
&= 0
\end{align}
where the relations $t_2 = +(t_1+t_3) + \phi/\omega_0$, and $\omega_1 = - \omega_0$ have been used. This solution is the one that produces the smallest control times, required for time optimality.  To proceed further, we make the following change of variables $t_1 = t-a$, $t_3=t+a$, and we have:
\begin{equation}
3 - 4 \cos(\omega_0 a) \cos(\omega_0 t) + 2 \cos(2 \omega_0 t) = 0
\end{equation}
\begin{equation}
\Rightarrow \omega_0 a = \pm \arccos\left( \frac{3 + 2 \cos(2 \omega_0 t) }{4 \cos(\omega_0 t)} \right) + 2 m \pi ~ ;~ m \in \setZ.
\label{eq:solution_of_a}
\end{equation}
We obtain $\omega_0 \alpha \in \setC$, except for particular points where it is real. The real values are: $\omega_0 a = 2 m' \pi$ with $m' \in \setZ$. Then, the solutions with the shortest control durations are given by $a=0$, or equivalently $\alpha = 1 \Rightarrow t_1=t_3$ . This symmetry has been observed in many different studies~\cite{KOBZAR2004,KOBZAR2008,KOBZAR2012,vandamme2017a}.

Robust $\pi$-pulses have been studied extensively~\cite{cat}. Here we compare our approach with the solutions of \cite{vandamme2017a}. We consider two optimal solutions of this study. The first one consists in a robust state-to-state transfer from the north pole of the Bloch sphere to the south pole, while the second one is a robust $SO(3)$-transformation of angle $\phi=\pi$. In order to compare the results, we use $\omega_0 = 1$. In the first case, the optimal solution is given by a $B-B$ control. In figure~\ref{fig:comparison_leo_2_switchings_compil}, we observe that the robustness is not improved by such a control. However this concerns the full rotation matrix. If we restrict the transformation to the $z$- axis only, we observe an enhancement of the robustness. The optimal solution is given by $k=1$ and $n=2$, which gives $t_1 = 3 \pi/2$ and $t_2 = \pi /2$, in agreement with the results obtained in \cite{vandamme2017a}. For the robust $\pi$-pulse, we have to consider the two-switching case. We found the optimal values $(n=5,k=3,\alpha=1)$, which correspond to the result obtained numerically in \cite{vandamme2017a}. Note that other similar solutions can be derived from this approach. In particular, we point out the solution $(n=3,k=2,\alpha=1)$ which is less robust around $\Delta = 0$, but more robust in average. This new control field is shorter than the reference one ($2\pi$ instead of $2.34\pi$). The fidelity of some solutions, and their corresponding control fields are displayed in figure~\ref{fig:comparison_leo_2_switchings_compil}.

The three-switching case has also been studied, and no better solution has been found. With four switchings and more, the number of cases becomes larger and additional constraints have to be accounted for in practical computations. A natural way to introduce such constraints is to incorporate other offsets in the optimization process.
For example, one could consider: $\Mc C_\Delta = \{ 0,\Delta_1/2,\Delta_1 \}$ or $\Mc C_\Delta = \{ 0,\Delta_1/3,2\Delta_1/3,\Delta_1 \}$,...
Note that a way to improve locally the robustness is to consider a selective control with a local robustness, as the one considered in \ref{sec:localy robust pulse}. The idea is to consider an offset $\Delta_2<\Delta_1$ different from $0$ that produces the rotation of angle $\phi$. If $\Delta_2$ is close to zero, the robustness on the interval $[-\Delta_1,\Delta_1]$ is very good, even if the rotation is not exact at resonance. With this system, computation can be performed in the two-switching case.
\begin{figure}
\includegraphics[width=\textwidth]{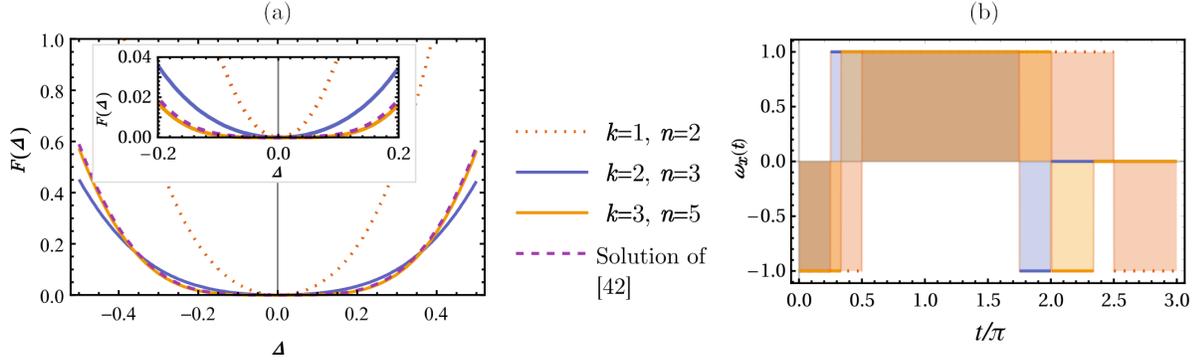}
\caption{(Color online) (a) Fidelity function $F$ for a $\pi$-pulse for different values of $(n,k)$, and for the solution of \cite{vandamme2017a}. The inset shows $F$ near $\Delta =0$. (b) Control field associated with each solution. Control times are respectively: $3\pi$, $2\pi$ , $2.34\pi$, $2.34\pi$ (top to bottom in the legend).}
\label{fig:comparison_leo_2_switchings_compil}
\end{figure}

\section{Conclusion}
\label{sec:conclusion}
We have presented different tools for the design of selective and robust $SO(3)$-transformations. Our computations are based on a model of two spins with different offsets. The first offset (at resonance) is supposed to realize the desired $SO(3)$-transformation, while the second offset produces the identity transformation. This allows us to derive a list of constraints on the control field that reduces considerably the number of variables required to parameterize the optimal solution. We prove different results supporting the optimality of singular trajectories in the selective case. Hence, we can conjecture the existence of a link between singular trajectories and selective controls and between regular trajectories and robust controls. The change of behavior between the two control problems occurs for $\Delta_1=\Delta_0$. This corresponds to the offset value for which the set of singular trajectories reaches the set of regular trajectories. This point is summarized in figure~\ref{fig:selective_robust_areas}.

We have found that time-optimal selective controls are given by constant controls of amplitude $|\omega_S |< \omega_0$, and time-optimal robust controls have been determined analytically by calculating the evolution operator up to two switchings.
Several known control protocols have been found and new solutions have been also highlighted.

The extension of such methods to more switchings or to the general case of two-input control fields represent the main perspective of this study.
The change of structure in the time-optimal solution for more than two switchings has been observed several times~\cite{KOBZAR2008,KOBZAR2012,van_damme_time-optimal_2018} (e.g. a transition from a bang-bang control to a smooth control with two inputs when the robustness increases). It could be interesting to explain analytically these numerical observations with the framework presented in this paper. Finally, we mention that calculations of second order variations in the sense of \cite{Stefani_2004} shall provide additional information on the optimality of singular or regular trajectories.

\begin{figure}[h]
\begin{center}
\includegraphics[width=\textwidth]{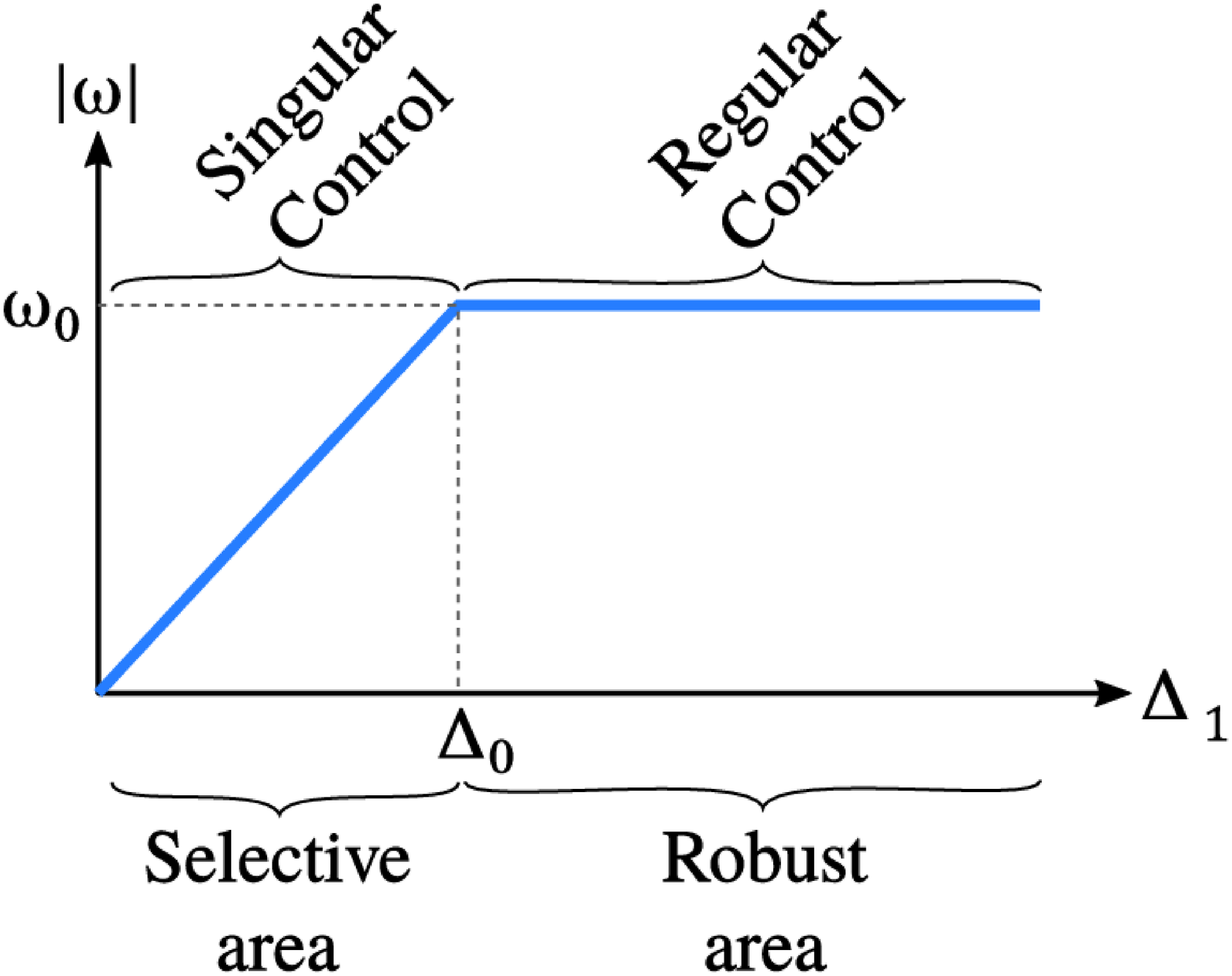}
\end{center}
\caption{Transition from the area of selective controls characterized by singular solutions of the PMP, to the area of robust controls characterized by regular trajectories.}
\label{fig:selective_robust_areas}
\end{figure}

\ack

This work was supported by the French ``Investissements d'Avenir'' program, project ISITE-BFC,emergent project I-QUINS (contract ANR-15-IDEX-03). This project has received funding from the European Union’s Horizon 2020 research and innovation programme under
the Marie-Sklodowska-Curie grant agreement No 765267 (QUSCO). S.Glaser acknowledges support from the Deutsche Forschungsgemeinschaft (DFG, German Research Foundation) under Germany’s Excellence Strategy – EXC-2111 – 390814868. The uthors  acknowledge L. Van Damme for helpful discussions.

\appendix

\section{Product of rotations with quaternions}
\label{sec:product of rotation}
We consider, in this paragraph, the product of two rotations with different axes. We refer to Ref.~\cite{kuipers_quaternions_2002,biedenharn_angular_1981} for further details.
Let $\alpha$ and $\beta$ be two rotation angles around the respective axes $\vec n_1$ and $\vec n_2$. The composition of the two rotations gives an effective angle of rotation $\gamma$ around an axis $\vec n_3$, which can be expressed as:
\begin{equation}
\cos \left(\frac{\gamma}{2}\right)  =  \cos\left( \frac{\alpha - \beta}{2}\right)\sin^2\left( \frac{\theta}{2}\right) + \cos\left( \frac{\alpha+\beta}{2}\right) \cos^2\left(\frac{\theta}{2}\right)
\end{equation}

\begin{equation}
\vec n_3 = \frac{ \left( \cos (\alpha/2) \sin (\beta/2) \vec n_2 + \cos (\beta/2) \sin (\alpha/2) \vec n_1 - \sin (\alpha/2) \sin (\beta/2) \vec n_1 \wedge \vec n_2 \right)}{\sin(\gamma/2)}
\end{equation}
where $\vec n_1.\vec n_2 = \cos(\theta)$. A detailed proof is given only for the computation of the rotation angle. The computation of the rotation axis can be done along the same line~\cite{kuipers_quaternions_2002,biedenharn_angular_1981}. For the sake of simplicity, we use the quaternion representation of rotations.  A rotation of angle $\alpha$ around an axis $\vec n$ is described by the quaternion $\mathbf{q} = \cos(\alpha/2) -\sin(\alpha/2) (n_x \mathbf{i}+n_y \mathbf{j}+ n_z \mathbf{z})$. We can only focus on the real part of the quaternion to determine the rotation angle. A straightforward computation gives:
\begin{equation*}
 \Re[\mathbf{q_1}\mathbf{q_2}] = \cos\left( \frac{\alpha}{2}\right)\cos\left( \frac{\beta}{2}\right)-\sin\left( \frac{\alpha}{2}\right)\sin\left( \frac{\beta}{2}\right) \vec n_1 . \vec n_2 ,
\end{equation*}
where $\Re[\cdot]$ is the real part of the quaternion. Using $\Re[\mathbf{q_1} \mathbf{q_2}]=\cos(\gamma/2)$, we obtain:
 \begin{equation*}
\cos\left( \frac{\gamma}{2}\right) = \cos\left( \frac{\alpha}{2}\right)\cos\left( \frac{\beta}{2}\right)-\sin\left( \frac{\alpha}{2}\right)\sin\left( \frac{\beta}{2}\right) \cos(\theta) ,
 \end{equation*}
and simple algebra leads to:
\begin{equation*}
\cos\left( \frac{\gamma}{2}\right) = \cos\left( \frac{\alpha - \beta}{2}\right)\sin^2\left( \frac{\theta}{2}\right) + \cos\left( \frac{\alpha+\beta}{2}\right) \cos^2\left(\frac{\theta}{2}\right) .
\end{equation*}

\section{Selective controls with local robustness}
\label{sec:localy robust pulse}

In this appendix, we briefly present the extension of the main model to two other situations for which the notion of selectivity and robustness are combined. Here, the goal is to produce a selective transformation, but with a small interval of robustness. This interval is defined by two neighborhood offsets associated with the same target state. Since the problem complexity increases quickly with $N$, we focus on numerical simulations.

For the first example, we choose the following ensemble of offsets:
\[
\Mc C_\Delta = \{-\Delta_2,-\Delta_1,\Delta_1,\Delta_2\},
\]
and the corresponding target transformations:
\[
\Mc C_{\hat U} = \{ \hat \Id,e^{\phi \hat \epsilon_x},e^{\phi \hat \epsilon_x},\hat \Id \}.
\]
The optimal field required to generate $e^{\phi \hat \epsilon_x}$ with $\Delta_1 \neq 0$ is a regular control with potentially several switchings~\cite{boscain-mason,boscainchitour}. For simplicity, we study here a solution with at most one switching (more switchings can also be considered, as in the main text). Hence, the control is characterized by two bangs of durations $t_1$ and $t_2$. We also define $T = t_1+t_2$. The phase of the control is chosen to be positive at the origin. Following the computations of section~\ref{sec:BB}, we found the conditions to produce the identity:
\begin{align}
\label{eq:Delta_opt_case20}
(t_1-t_2) \sqrt{\omega_0^2 + \Delta_2^2} & = 2 k \pi \\
\label{eq:Delta_opt_case2}
(t_1+t_2) \sqrt{\omega_0^2 + \Delta_2^2}& = 2 n \pi .
\end{align}
From a given value of $\Delta_1$, we get $t_1$ and $t_2$, the ratio $k/n$ and $\Delta_2$. The optimality of the pulse sequence is obtained by using the same method as in section~\ref{sec:numerical verification}. The cost function $C$ is minimized numerically with GRAPE~\cite{KHANEJA2005} for a wide range of system parameters and the results are compared with analytic computations. Results are displayed in figure~\ref{fig:landscape_case2}. We observe that \eqref{eq:Delta_opt_case20} and \eqref{eq:Delta_opt_case2} describe the minima of $C$.
\begin{figure}[h]
\begin{center}
\includegraphics[width=\textwidth]{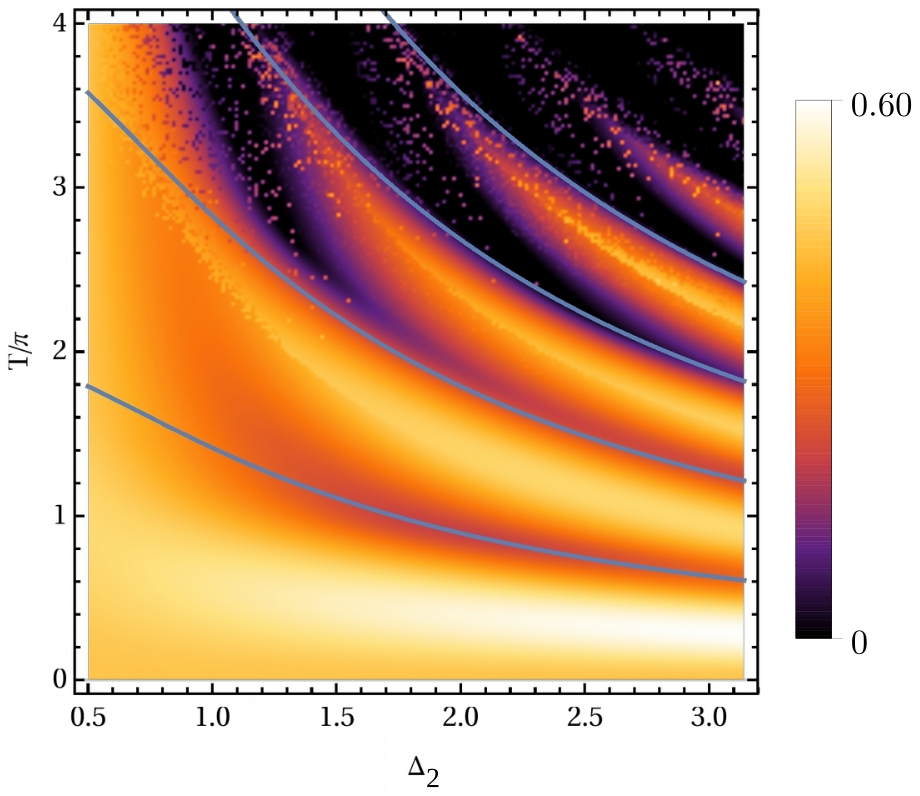}
\end{center}
\caption{(Color online) Plot of $C$ as a function of the control duration $T$ and the offset $\Delta_2$, for a locally-broadband selective rotation of angle $\phi=\pi$. Each point of the contour plot corresponds to an optimization. Blue curves are given by \eqref{eq:Delta_opt_case2} for different values of $k$. Parameters are set to $\Delta_1 = 0.5$, $\omega_0=1$ and $\phi = \pi$.}
\label{fig:landscape_case2}
\end{figure}

\begin{figure}[h]
\begin{center}
\includegraphics[width=\textwidth]{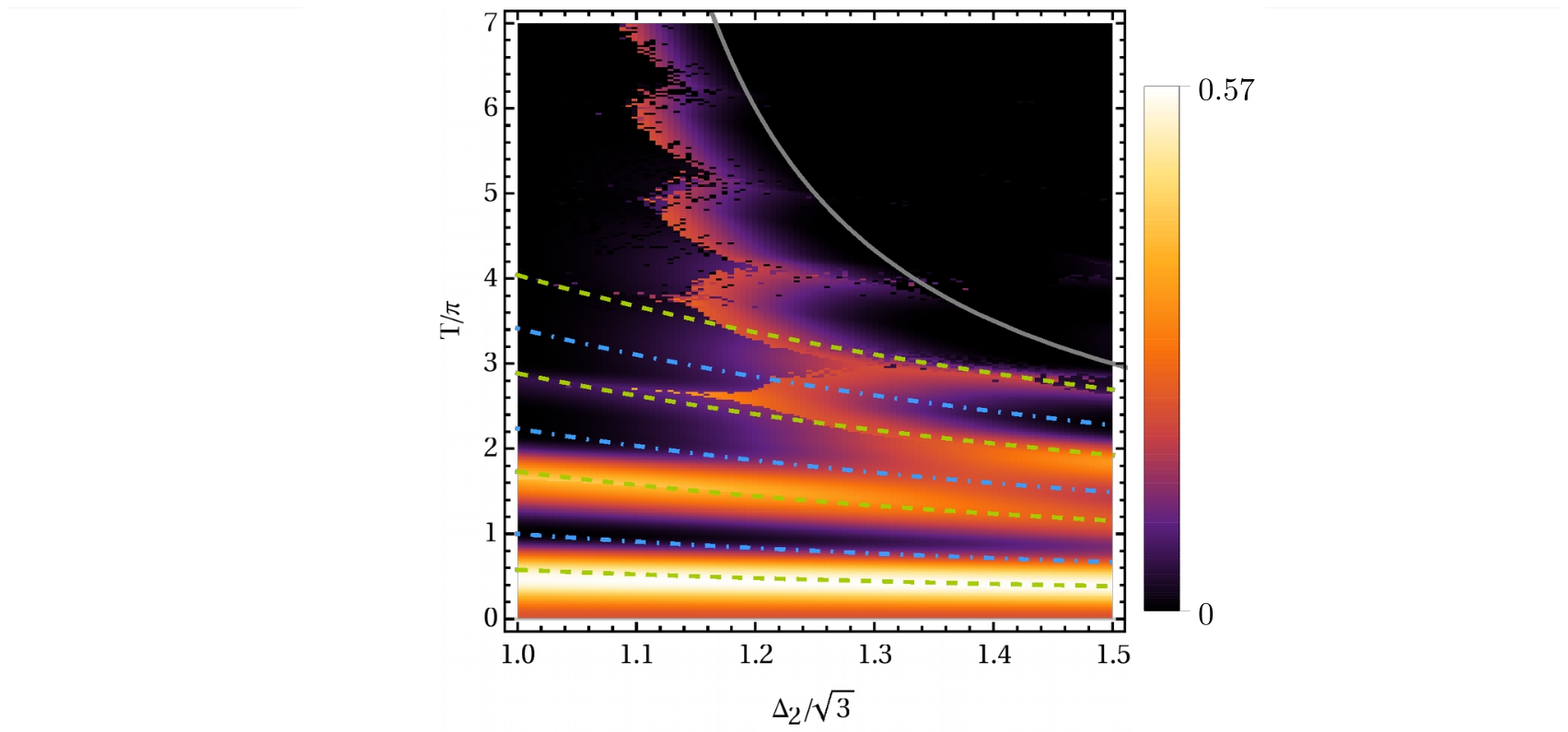}
\end{center}
\caption{(Color online) Same as Fig.~\ref{fig:landscape_case2}, but for the second example of this appendix. $\Delta_1$ is set to $\Delta_1=\sqrt{3}$. Blue dot-dashed curves are given by \eqref{eq:time_locally_broadband_pulse_2} for different values of $k$, while green dashed curves correspond to \eqref{eq:phase_switch_optimal_control_1_offset}. The gray curve is given by \eqref{eq:limit_Delta_2_case_3}.}\label{fig:landscape_case3}
\end{figure}

The second example is defined by:
\[
\mathcal{C}_\Delta = \{ -\Delta_2,-\Delta_1,0,\Delta_1,\Delta_2\} ,
\]
\[
\mathcal{C}_{\hat U} = \{ \hat \Id,\hat \Id, e^{\phi \hat \epsilon_x},\hat \Id,\hat \Id \}.
\]
There is a trivial solution for $\Delta = \sqrt{(2k \pi)^2 - \phi ^2}/T$, so we impose that $T |\Delta_2| \in [\sqrt{(2 \pi)^2 - \phi ^2},\sqrt{(4 \pi)^2 - \phi ^2}]$. In this case, there is no analytic solution and we start directly with a numerical optimization. The minimum of $C$ as a function of $\Delta_2$ is plotted in figure~\ref{fig:landscape_case3}, with $\omega_0=1$, $\Delta_1 = \sqrt{3}$, and $\phi = \pi$. We observe a time-optimal upper bound approximately given by:
\begin{equation}
T = \sqrt{4\pi^2 - \phi^2}\left(\frac{1}{\Delta_1}+ \frac{1}{\Delta_2 - \Delta_1}\right).
\label{eq:limit_Delta_2_case_3}
\end{equation}
This equation is deduced heuristically from the following assumption. The time required by the full transformation is given by the minimum time to reach the target state for $\Delta_1$ (see \eqref{eq:Delta_singular}), plus the one to get the target state for $\Delta_2$, in the rotating frame centered in $\Delta_1$. This time is also given by \eqref{eq:Delta_singular}, but the offset is replaced by the offset difference $\Delta_2-\Delta_1$.

Interesting solutions are also found below the curve of \eqref{eq:limit_Delta_2_case_3}. The control landscape has a complex structure with areas of extremely small costs ($<10^{-13}$) surrounded by areas of extremely large costs. Additionally, the transition between such areas is generally not smooth. Areas of high and low fidelities can be described qualitatively by two functions. The first one is:
\begin{equation}
T = \frac{1}{\Delta_2} \sqrt{4 k^2 \pi^2 -\phi ^2},
\label{eq:time_locally_broadband_pulse_2}
\end{equation}
while the second one is the control phase switching in the optimal control of a single offset (see proposition 6 in \cite{garon2013}):
\begin{equation}
T = -\frac{(4 n \pm 1) \pi}{\Delta_2}.
\label{eq:phase_switch_optimal_control_1_offset}
\end{equation}
These two equations approximate respectively low and high values of $C$ near $\Delta_2 \simeq \Delta_1$, and near \eqref{eq:limit_Delta_2_case_3}.

We also observe other discontinuities in the landscape, which are not fully understood yet. From numerical observation, it seems that the discontinuity occurs by a change of the optimal field from singular to regular extremal trajectories, and reciprocally. Further investigations are required to confirm this point.

\section*{References}
\bibliographystyle{vancouver}
\bibliography{bib_file}

\begin{thebibliography}{10}

\bibitem{cat}
{Glaser} SJ, {Boscain} U, {Calarco} T, {Koch} CP, {K{\"o}ckenberger} W,
  {Kosloff} R, et~al.
\newblock {Training Schr{\"o}dinger's cat: quantum optimal control. Strategic
  report on current status, visions and goals for research in Europe}.
\newblock European Physical Journal D. 2015;69:279.

\bibitem{past-present-future}
Brif C, Chakrabarti R, Rabitz H.
\newblock Control of quantum phenomena: past, present and future.
\newblock New Journal of Physics. 2010;12(7):075008.

\bibitem{dong}
Dong D, Petersen IR.
\newblock Quantum control theory and applications: a survey.
\newblock IET Control Theory \& Applications. 2010;4(12):2651--2671.

\bibitem{altafini-ticozzi}
Altafini C, Ticozzi F.
\newblock Modeling and control of quantum systems: an introduction.
\newblock IEEE Trans Automat Control. 2012;57(8):1898--1917.

\bibitem{RMPsugny}
Koch CP, Lemeshko M, Sugny D.
\newblock Quantum control of molecular rotation.
\newblock Rev Mod Phys. 2019 Sep;91:035005.

\bibitem{bonnard_optimal_2012}
Bonnard B, Sugny D.
\newblock Optimal {Control} with {Applications} in {Space} and {Quantum}
  {Dynamics}.
\newblock American Institute of Mathematical Sciences; 2012.

\bibitem{pontryaginbook}
Pontryagin LS, Boltianski V, Gamkrelidze R, Mitchtchenko E.
\newblock {The Mathematical Theory of Optimal Processes}.
\newblock {John Wiley and Sons, New York}; 1962.

\bibitem{liberzon-book}
Liberzon D.
\newblock Calculus of variations and optimal control theory.
\newblock Princeton University Press, Princeton, NJ; 2012.
\newblock A concise introduction.

\bibitem{brysonbook}
Bryson AE, Ho YC.
\newblock {Applied Optimal Control: Optimization, Estimation, and Control}.
\newblock {Taylor and Francis, Philadelphia}; 1975.

\bibitem{leemarkusbook}
Lee MM, Markus L.
\newblock {Foundations of Optimal Control Theory}.
\newblock {John Wiley and Sons, New York}; 1967.

\bibitem{dalessandro-book}
D'Alessandro D.
\newblock {Introduction to quantum control and dynamics.}
\newblock {Applied Mathematics and Nonlinear Science Series. Boca Raton, FL:
  Chapman, Hall/CRC.}; 2008.

\bibitem{kirk_optimal_2004}
Kirk DE.
\newblock Optimal {Control} {Theory}: {An} {Introduction}.
\newblock Courier Corporation; 2004.

\bibitem{borzi-book}
Borzi A, Ciaramella G, Sprengel M.
\newblock Formulation and numerical solution of quantum control problems.
  vol.~16 of Computational Science \& Engineering.
\newblock Society for Industrial and Applied Mathematics (SIAM), Philadelphia,
  PA; 2017.

\bibitem{boscain-book}
Boscain U, Piccoli B.
\newblock Optimal syntheses for control systems on 2-{D} manifolds. vol.~43 of
  Math\'{e}matiques \& Applications (Berlin) [Mathematics \& Applications].
\newblock Springer-Verlag, Berlin; 2004.

\bibitem{jurdjevic-book}
Jurdjevic V.
\newblock Geometric control theory. vol.~52 of Cambridge Studies in Advanced
  Mathematics.
\newblock Cambridge University Press, Cambridge; 1997.

\bibitem{roadmapQT}
Acin A, Bloch I, Buhrman H, Calarco T, Eichler C, Eisert J, et~al.
\newblock The quantum technologies roadmap: a European community view.
\newblock New Journal of Physics. 2018 aug;20(8):080201.

\bibitem{boscain-mason}
Boscain U, Mason P.
\newblock Time minimal trajectories for a spin {$1/2$} particle in a magnetic
  field.
\newblock J Math Phys. 2006;47(6):062101, 29.

\bibitem{alessandro2001}
{D'Alessandro} D, {Dahleh} M.
\newblock Optimal control of two-level quantum systems.
\newblock IEEE Transactions on Automatic Control. 2001;46(6):866--876.

\bibitem{sugny10}
Ass\'emat E, Lapert M, Zhang Y, Braun M, Glaser SJ, Sugny D.
\newblock Simultaneous time-optimal control of the inversion of two
  spin-$\frac{1}{2}$ particles.
\newblock Phys Rev A. 2010 Jul;82:013415.

\bibitem{hegerfeldt2013}
Hegerfeldt GC.
\newblock Driving at the Quantum Speed Limit: Optimal Control of a Two-Level
  System.
\newblock Phys Rev Lett. 2013 Dec;111:260501.

\bibitem{lapert2010}
Lapert M, Zhang Y, Braun M, Glaser SJ, Sugny D.
\newblock Singular Extremals for the Time-Optimal Control of Dissipative Spin
  $\frac{1}{2}$ Particles.
\newblock Phys Rev Lett. 2010 Feb;104:083001.

\bibitem{bonnard2012}
{Bonnard} B, {Cots} O, {Glaser} SJ, {Lapert} M, {Sugny} D, {Zhang} Y.
\newblock Geometric Optimal Control of the Contrast Imaging Problem in Nuclear
  Magnetic Resonance.
\newblock IEEE Transactions on Automatic Control. 2012;57(8):1957--1969.

\bibitem{KhanejaPNAS}
Khaneja N, Luy B, Glaser SJ.
\newblock Boundary of quantum evolution under decoherence.
\newblock Proceedings of the National Academy of Sciences.
  2003;100(23):13162--13166.

\bibitem{khanejathree}
Sklarz SE, Tannor DJ, Khaneja N.
\newblock Optimal control of quantum dissipative dynamics: Analytic solution
  for cooling the three-level $\ensuremath{\Lambda}$ system.
\newblock Phys Rev A. 2004 May;69:053408.

\bibitem{albertini_minimum_2014}
Albertini F, D'Alessandro D.
\newblock Minimum {Time} {Optimal} {Synthesis} for a {Control} {System} on
  {SU}(2).
\newblock arXiv:14077491 [quant-ph]. 2014 Jul.
\newblock ArXiv: 1407.7491.

\bibitem{garon2013}
Garon A, Glaser SJ, Sugny D.
\newblock Time-optimal control of SU(2) quantum operations.
\newblock Phys Rev A. 2013 Oct;88:043422.

\bibitem{khaneja2001}
Khaneja N, Brockett R, Glaser SJ.
\newblock Time optimal control in spin systems.
\newblock Phys Rev A. 2001 Feb;63:032308.

\bibitem{khaneja2002}
Khaneja N, Glaser SJ, Brockett R.
\newblock Sub-Riemannian geometry and time optimal control of three spin
  systems: Quantum gates and coherence transfer.
\newblock Phys Rev A. 2002 Jan;65:032301.

\bibitem{koch2012}
Reich DM, Ndong M, Koch CP.
\newblock Monotonically convergent optimization in quantum control using
  Krotov's method.
\newblock The Journal of Chemical Physics. 2012;136(10):104103.

\bibitem{KHANEJA2005}
Khaneja N, Reiss T, Kehlet C, Schulte-Herbrüggen T, Glaser SJ.
\newblock Optimal control of coupled spin dynamics: design of NMR pulse
  sequences by gradient ascent algorithms.
\newblock Journal of Magnetic Resonance. 2005;172(2):296 -- 305.

\bibitem{calarco2011}
Doria P, Calarco T, Montangero S.
\newblock Optimal Control Technique for Many-Body Quantum Dynamics.
\newblock Phys Rev Lett. 2011 May;106:190501.

\bibitem{k1}
Li JS, Khaneja N.
\newblock Control of inhomogeneous quantum ensembles.
\newblock Phys Rev A. 2006;73:030302.

\bibitem{k2}
Li JS, Khaneja N.
\newblock Ensemble control of {B}loch equations.
\newblock IEEE Trans Automat Control. 2009;54(3):528--536.

\bibitem{KOBZAR2004}
Kobzar K, Skinner TE, Khaneja N, Glaser SJ, Luy B.
\newblock Exploring the limits of broadband excitation and inversion pulses.
\newblock Journal of Magnetic Resonance. 2004;170(2):236 -- 243.

\bibitem{KOBZAR2005}
Kobzar K, Luy B, Khaneja N, Glaser SJ.
\newblock Pattern pulses: design of arbitrary excitation profiles as a function
  of pulse amplitude and offset.
\newblock Journal of Magnetic Resonance. 2005;173(2):229 -- 235.

\bibitem{KOBZAR2008}
Kobzar K, Skinner TE, Khaneja N, Glaser SJ, Luy B.
\newblock Exploring the limits of broadband excitation and inversion: II.
  Rf-power optimized pulses.
\newblock Journal of Magnetic Resonance. 2008;194(1):58 -- 66.

\bibitem{KOBZAR2012}
Kobzar K, Ehni S, Skinner TE, Glaser SJ, Luy B.
\newblock Exploring the limits of broadband 90° and 180° universal rotation
  pulses.
\newblock Journal of Magnetic Resonance. 2012;225:142 -- 160.

\bibitem{SKINNER2012}
Skinner TE, Gershenzon NI, Nimbalkar M, Glaser SJ.
\newblock Optimal control design of band-selective excitation pulses that
  accommodate relaxation and RF inhomogeneity.
\newblock Journal of Magnetic Resonance. 2012;217:53 -- 60.

\bibitem{rabitz14}
Chen C, Dong D, Long R, Petersen IR, Rabitz HA.
\newblock Sampling-based learning control of inhomogeneous quantum ensembles.
\newblock Phys Rev A. 2014 Feb;89:023402.

\bibitem{turinici19}
Turinici G.
\newblock Stochastic learning control of inhomogeneous quantum ensembles.
\newblock Phys Rev A. 2019 Nov;100:053403.

\bibitem{van_damme_time-optimal_2018}
Van~Damme L, Ansel Q, Glaser SJ, Sugny D.
\newblock Time-optimal selective pulses of two uncoupled spin-1/2 particles.
\newblock Phys Rev A. 2018 Oct;98(4):043421.
\newblock Publisher: American Physical Society.

\bibitem{vandamme2017a}
Van~Damme L, Ansel Q, Glaser SJ, Sugny D.
\newblock Robust optimal control of two-level quantum systems.
\newblock Phys Rev A. 2017 Jun;95:063403.

\bibitem{vandamme2017b}
Van-Damme L, Schraft D, Genov GT, Sugny D, Halfmann T, Gu\'erin S.
\newblock Robust not gate by single-shot-shaped pulses: Demonstration of the
  efficiency of the pulses in rephasing atomic coherences.
\newblock Phys Rev A. 2017 Aug;96:022309.

\bibitem{daems:2013}
Daems D, Ruschhaupt A, Sugny D, Gu\'erin S.
\newblock Robust Quantum Control by a Single-Shot Shaped Pulse.
\newblock Phys Rev Lett. 2013 Jul;111:050404.

\bibitem{barnes2019}
Zeng J, Yang CH, Dzurak AS, Barnes E.
\newblock Geometric formalism for constructing arbitrary single-qubit
  dynamically corrected gates.
\newblock Phys Rev A. 2019 May;99:052321.

\bibitem{barnes2018}
Zeng J, Barnes E.
\newblock Fastest pulses that implement dynamically corrected single-qubit
  phase gates.
\newblock Phys Rev A. 2018 Jul;98:012301.

\bibitem{vitanov2014}
Genov GT, Schraft D, Halfmann T, Vitanov NV.
\newblock Correction of Arbitrary Field Errors in Population Inversion of
  Quantum Systems by Universal Composite Pulses.
\newblock Phys Rev Lett. 2014 Jul;113:043001.

\bibitem{Ruschhaupt_2012}
Ruschhaupt A, Chen X, Alonso D, Muga JG.
\newblock Optimally robust shortcuts to population inversion in two-level
  quantum systems.
\newblock New Journal of Physics. 2012 sep;14(9):093040.

\bibitem{Zeng_2018}
Zeng J, Deng XH, Russo A, Barnes E.
\newblock General solution to inhomogeneous dephasing and smooth pulse
  dynamical decoupling.
\newblock New Journal of Physics. 2018 mar;20(3):033011.

\bibitem{linature}
Li JS, Ruths J, Glaser SJ.
\newblock Exact broadband excitation of two-level systems by mapping spins to
  springs.
\newblock Nat Comm. 2017;8:446.

\bibitem{ernstbook}
Ernst RR, Bodenhausen G, Wokaun A.
\newblock {Principles of Nuclear Magnetic Resonance in One and Two Dimensions}.
\newblock {Clarendon Press, Oxford}; 1987.

\bibitem{levittbook}
Levitt MH.
\newblock {Spin Dynamics: Basics of Nuclear Magnetic Resonance}.
\newblock {Wiley, New York}; 2008.

\bibitem{silver1985}
Silver MS, Joseph RI, Hoult DI.
\newblock Selective spin inversion in nuclear magnetic resonance and coherent
  optics through an exact solution of the Bloch-Riccati equation.
\newblock Phys Rev A. 1985 Apr;31:2753--2755.

\bibitem{silvernature}
Silver MS, Joseph RI, Chen CN, Sank VJ, Hoult DI.
\newblock Selective population inversion in NMR.
\newblock Nature. 1984;310:681.

\bibitem{LEVITT198661}
Levitt MH.
\newblock Composite pulses.
\newblock Progress in Nuclear Magnetic Resonance Spectroscopy. 1986;18(2):61 --
  122.

\bibitem{hall_lie_2015}
Hall BC.
\newblock Lie {Groups}, {Lie} {Algebras}, and {Representations}: {An}
  {Elementary} {Introduction}.
\newblock 2nd ed. Graduate {Texts} in {Mathematics}. Springer International
  Publishing; 2015.

\bibitem{lapert_understanding_2013}
Lapert M, Assémat E, Glaser SJ, Sugny D.
\newblock Understanding the global structure of two-level quantum systems with
  relaxation: Vector fields organized through the magic plane and the
  steady-state ellipsoid.
\newblock Physical Review A. 2013 Sep;88(3):033407.
\newblock Publisher: American Physical Society.

\bibitem{bernstienbook}
Bernstein MA, King KF, Zhou XJ.
\newblock {Handbook of MRI pulse sequenques}.
\newblock {Elsevier Academic Press, New York}; 2004.

\bibitem{brinkmann_introduction_2018}
Brinkmann A.
\newblock Introduction to average {Hamiltonian} theory. {I}. {Basics}.
\newblock Concepts in Magnetic Resonance Part A. 2016.

\bibitem{suzuki_generalized_1976}
Suzuki M.
\newblock Generalized {Trotter}'s formula and systematic approximants of
  exponential operators and inner derivations with applications to many-body
  problems.
\newblock Communications in Mathematical Physics. 1976;51(2):183--190.

\bibitem{Stefani_2004}
{Stefani} G.
\newblock Minimum-time optimality of a singular arc: second order sufficient
  conditions.
\newblock In: 2004 43rd IEEE Conference on Decision and Control (CDC) (IEEE
  Cat. No.04CH37601). vol.~1; 2004. p. 450--454 Vol.1.

\bibitem{kuipers_quaternions_2002}
Kuipers JB.
\newblock Quaternions and {Rotation} {Sequences}: {A} {Primer} with
  {Applications} to {Orbits}, {Aerospace} and {Virtual} {Reality}.
\newblock Princeton, N.J: Princeton University Press; 2002.

\bibitem{biedenharn_angular_1981}
Biedenharn LC.
\newblock Angular {Momentum} in {Quantum} {Physics}: {Theory} and
  {Application}.
\newblock Addison-Wesley Publishing Company, Advanced Book Program; 1981.

\bibitem{boscainchitour}
Boscain U, Chitour Y.
\newblock Time Optimal Synthesis for Left–Invariant Control Systems on SO(3).
\newblock SIAM SICON, Journal on Control and Optimization. 2005;44:111.

\end{thebibliography}

\end{document}